\documentclass[11pt,a4paper]{article}
\usepackage{epsf,epsfig,amsfonts,amsgen,amsmath,amstext,amsbsy,amsopn,amsthm}
\usepackage{amssymb}
\usepackage{cite}
\usepackage{enumitem}
\newtheorem{theorem}{Theorem}

\theoremstyle{definition}

\newtheorem{example}{Example}

\begin{document}
\title{\bf {\Large   Geometry of the symplectic group and optimal EAQECC codes}}
\date{}
\author{ Ruihu Li$^{1}$, Yuezhen Ren$^{1}$, Chaofeng Guan $^{1}$and  Yang Liu$^{1*}$\\
1. Department of Basic Science, Air Force Engineering University,
Xi'an, Shaanxi.\\(e-mail: liruihu@aliyun.com;
renyzlw@163.com;liu$\_$yang10@163.com).\\  2. Henan Key Laboratory
of Network Cryptography
Technology, Zhengzhou, Henan.\\(e-mail:gcf2020@yeah.net)\\
} \maketitle
\begin{abstract}
A new link between the geometry of the symplectic group and
entanglement-assisted (EA) quantum error-correcting codes (EAQECCs)
is presented. Relations between symplectic subspaces and quaternary
additive codes concerning the parameters of EAQECCs are described.
Consequently, parameters of EA stabilizer codes are revealed within
the framework of additive codes. Our techniques enable us to solve
some open problems regarding optimal EAQECCs and
entanglement-assisted quantum minimum distance separable (EAQMDS)
codes, and are also useful for designing encoding and decoding
quantum circuits for EA stabilizer codes.


\end{abstract}

\section{Introduction}
Quantum-error correcting codes (QECCs) can correct
 errors in quantum
communication and quantum computation, and are an indispensable
ingredient for quantum information processing. Since the pioneering
work of Shor and Steane \cite{shor1995scheme, steane1996error},
researchers have focused on finding optimal QECCs. The most studied
class of QECCs are stabilizer quantum codes, also called additive
QECCs or standard quantum codes. Such codes can be constructed from
classical additive codes or linear codes satisfying certain
self-orthogonal properties \cite{gottesman1997stabilizer,
calderbank1998quantum, ketkar2006nonbinary}. Using this constructive
method, a large number of QECCs with suitable parameters have been
obtained, and Grassl et al. summarized these results and established
an online code table of the best-known binary QECCs
\cite{grassl2006bounds}. The self-orthogonal properties form a
barrier to incorporating all classical codes into QECCs
\cite{mackay2004sparse, bowen2002entanglement, brun2006correcting}.

In \cite{brun2006correcting}, Brun, Devetak, and Hsieh devised the
entanglement-assisted (EA) stabilizer formalism, which includes the
standard stabilizer formalism in \cite{gottesman1997stabilizer,
calderbank1998quantum} as a special case. They showed that if shared
entanglement between the encoder and decoder is available, classical
linear quaternary (and binary) codes that are not self-orthogonal
can also be transformed into EAQECCs. EAQECCs constructed via this
EA-stabilizer formalism are named as EA stabilizer codes or additive
EAQECCs. Following \cite{brun2006correcting}, extensive research has
been conducted on constructing additive EAQECCs, optimizing their
parameters, and determining their bounds in \cite{wilde2008optimal,
lai2013entanglement, hsieh2011high, li2014entanglement,
lai2013duality, lai2014dualities, lai2018linear, guo2013linear,
grassl2021entanglement, grassl2022entropic,grassl2022bounds,
luo2022much, huber2020quantum, galindo2019entanglement}.

Known results provide evidence that entanglement enhances the
error-correcting ability of quantum codes \cite{lai2013entanglement}
and confirm the advantage of EA quantum LDPC codes over standard
quantum LDPC codes \cite{hsieh2011high}. Reference
\cite{li2014entanglement} shows that there are infinitely many
impure EAQECCs violating the EA-quantum Hamming bound. Grassl
\cite{grassl2021entanglement} proves that certain types of EAQECCs
violate the EA-quantum Singleton bound obtained in
\cite{brun2006correcting}. References \cite{lai2013duality,
lai2014dualities, lai2018linear, guo2013linear,
grassl2021entanglement, grassl2022entropic} have established
additional bounds for EAQECCs, generalized the EA-quantum Singleton
bound from \cite{brun2006correcting} in various ways, and proposed
open problems regarding optimal EAQECCs and EAQMDS codes. Two of
these problems are as follows:
\begin{itemize}
    \item[(a)] How to determine the optimality of an \([[n, k, d ; c]]_{q}\) code? Currently, even for \(n=5\) and \(q=2\), the minimum distance of the optimal \([[5,2, d ; 3]]_{2}\) EAQECC has not been determined \cite{grassl2022bounds}.
    \item[(b)] What constraints exist on the alphabet size \(q\) for the existence of an EAQMDS \([[n, k, d]]_{q}\) code? A similar issue for QMDS codes was addressed in \cite{huber2020quantum}.
\end{itemize}

It is therefore natural to consider new theories and techniques to
describe EAQECCs, discuss their constructions, and analyze their
optimality \cite{luo2022much, huber2020quantum,
galindo2019entanglement}. Reference \cite{lai2014dualities}
constructed EAQMDS codes \([[n, 1, n ; n-1]]\) for odd \(n\) and
showed that such codes do not exist for even \(n\). In this paper,
we attempt to solve these two open problems using the link between
the geometry of the symplectic group \cite{wan1993geometry} and
EAQECCs:
\begin{enumerate}
    \item First, we characterize the parameters of binary EAQECCs and provide methods to solve Problem (a).
    \item Second, we construct several classes of EAQMDS codes, which partially answer Problem (b).
\end{enumerate}

For QMDS codes, it is known that for \(d \geq 3\), the alphabet size
must satisfy \(n \leq q^2 + d - 2\). Researchers naturally expect a
similar constraint to hold for EAQMDS codes
\cite{grassl2022entropic}. However, we will show that this
conjecture is incorrect, as \(n\) can actually be arbitrarily large
for \(q=2\). To solve Problems (a) and (b), we construct many good
EAQECCs. Here, we only consider binary EA stabilizer codes; for
non-binary EAQECCs, refer to \cite{luo2022much, huber2020quantum,
galindo2019entanglement}.

\section{Symplectic Space, Additive Code and EAQECC}
 Let
\(F_{2}^{2n}\) be the \(2n\)-dimensional binary row vector space
over the binary field \(F_{2} = \{0,1\}\), whose elements are
denoted as \((a | b) = (a_{1}, a_{2}, \cdots, a_{n} | b_{1}, b_{2},
\cdots, b_{n})\). The symplectic weight \(wt_{s}(a | b)\) of \((a |
b)\) is defined as the number of coordinates \(i\) such that at
least one of \(a_{i}\) and \(b_{i}\) is 1, and the symplectic
distance \(d_{s}((a | b),(a' | b'))\) between \((a | b)\) and \((a'
| b')\) is defined as \(wt_{s}(a-a' | b-b')\). Let
\[K_{2n} = \begin{pmatrix} 0 & I_{n} \\ I_{n} & 0 \end{pmatrix}.\]

The symplectic inner product of \((a|b)\) and \((a' | b')\) with
respect to \(K_{2n}\) is defined as \(((a | b),(a' | b'))_{s} = (a |
b) K_{2n} (a' | b')^{T} = a(b')^{T} - b(a')^{T} = a(b')^{T} +
b(a')^{T}\). The space \(F_{2}^{2n}\) equipped with this symplectic
inner product is called a \(2n\)-dimensional symplectic space.

For a subspace \(V\) of \(F_{2}^{2n}\), the symplectic dual
\(V^{\perp_{s}}\) of \(V\) is
\[V^{\perp_{s}} = \left\{(x | y) \in F_{2}^{2n} \mid ((a | b),(x | y))_{s} = 0 \text{ for all } (a | b) \in V\right\}.\]

A subspace \(V\) of \(F_{2}^{2n}\) is called totally isotropic if
\(V \subseteq V^{\perp_{s}}\). Let \(P_{V}\) be a generator matrix
of \(V\) with dimension \(m\). If the rank of the matrix \(P_{V}
K_{2n} P_{V}^{T}\) is \(2c\), \(V\) is called a subspace of type
\((m, c)_{[s]}\). An \((m, c)_{[s]}\) subspace exists in
\(F_{2}^{2n}\) if and only if \(2c \leq m \leq n + c\), and the dual
space of an \((m, c)_{[s]}\) subspace is of type \((2n - m, n + c -
m)_{[s]}\). In particular, a subspace of type \((m, 0)_{[s]}\) is an
\(m\)-dimensional totally isotropic subspace, and the dual space of
a type \((m, 0)_{[s]}\) subspace is of type \((2n - m, n -
m)_{[s]}\) \cite{wan1993geometry}. A subspace of type \((2m,
m)_{[s]}\) in \(F_{2}^{2n}\) is called a \(2m\)-dimensional totally
non-isotropic subspace \cite{wan1993geometry} and a symplectic
subspace in \cite{brun2006correcting}.

Given a subspace \(V\) of \(F_{2}^{2n}\), we define \(R(V) = V \cap
V^{\perp_{s}}\) as the symplectic radical of \(V\) (or
\(V^{\perp_{s}}\)); its dimension \(l = \dim R(V)\) is called the
radical dimension of \(V\). If \(V\) is of type \((m, c)_{[s]}\),
its radical dimension is \(l = m - 2c\), hence \(c = \frac{m -
l}{2}\).

Let \(F_{4} = \{0, 1, \omega, \varpi\}\) be the four-element Galois
field where \(\varpi = 1 + \omega = \omega^{2}\) and \(\omega^{3} =
1\). The conjugation of \(x \in F_{4}\) is \(\bar{x} = x^{2}\), and
the conjugate transpose of a matrix \(G\) over \(F_{4}\) is
\((\bar{G})^{T} = G^{\dagger}\). The Hermitian inner product and
trace inner product of \(u, v \in F_{4}^{n}\) are defined as \((u,
v)_{h} = u \bar{v}^{T} = \sum_{j=1}^{n} u_{j} v_{j}^{2}\) and \((u,
v)_{t} = \text{tr}(u \bar{v}^{T}) = \sum_{j=1}^{n}(u_{j}
\overline{v_{j}} + \overline{u_{j}} v_{j}) = \sum_{j=1}^{n}(u_{j}
v_{j}^{2} + u_{j}^{2} v_{j})\), respectively
\cite{calderbank1998quantum}.

An additive subgroup \(C\) of \(F_{4}^{n}\) is called an additive
code over \(F_{4}\). If \(C\) contains \(2^{m}\) vectors, \(C\) is
denoted as an \((n, 2^{m})\) additive code. A matrix whose rows form
a basis of \(C\) over \(F_{2}\) is called an additive generator
matrix of \(C\). For \(C = (n, 2^{m})_{4}\), its trace dual code is
defined as
\[\mathcal{C}^{\perp_{t}} = \left\{u \in F_{4}^{n} \mid (u, v)_{t} = 0 \text{ for all } v \in \mathcal{C}\right\}.\]

\(C^{\perp_{t}}\) is an \((n, 2^{2n - m})\) additive code
\cite{calderbank1998quantum}, and a generator matrix of
\(C^{\perp_{t}}\) is called an additive parity check matrix of
\(C\). \(C\) is trace self-orthogonal if \(C \subseteq
C^{\perp_{t}}\).

We define an isometric map \(\phi\) from \(F_{2}^{2n}\) to
\(F_{4}^{n}\) as in \cite{calderbank1998quantum}, where \(\phi((a |
b)) = \omega a + \varpi b \in F_{4}^{n}\) for \(v = (a | b) \in
F_{2}^{2n}\). For any subspace \(S\) of \(F_{2}^{2n}\), \(\phi(S)\)
is an additive group (generally, it is not a subspace of
\(F_{4}^{n}\) but a vector space over \(F_{2}\)). If \(S\) is an
\(m\)-dimensional subspace of \(F_{2}^{2n}\), \(C = \phi(S)\) is an
\((n, 2^{m})\) additive code. For a generator matrix \(P_{S}\) of
\(S\), \(G = \phi(P_{S})\) is a generator matrix of \(C\) and \(G
G^{\dagger} + \overline{G G^{\dagger}} = P_{S} K_{2n} P_{S}^{T}\).
If \(S\) is of type \((m, c)_{[s]}\), then \(C\) is said to be of
type \((m, c)_{[t]}\). Thus, its trace radical \(R(C) = C \cap
C^{\perp_{t}}\) has \(\dim R(C) = l = m - 2c\), and
\(C^{\perp_{t}}\) is of type \((2n - m, n + c - m)_{[t]}\). A code
\(C\) of type \((m, 0)_{[t]}\) is trace self-orthogonal. If \(C\) is
an \([n, k]_{4}\) quaternary linear code with generator matrix \(G\)
and \(G G^{\dagger}\) has rank \(e\), then \(C\) is an \((n,
2^{2k})\) additive code and of type \((m, c)_{[t]} = (2k, e)_{[t]}\)
according to \cite{calderbank1998quantum, wilde2008optimal}.

Suppose \(G_{n}\) is the \(n\)-fold Pauli group and \(G_{n} =
\widehat{G}_{n} / \{i^{e} I, e = 0, 1, 2, 3\}\). Let \(S\) be a
subgroup of \(G_{n}\), and \(N(S)\) be the normalizer of \(S\) in
\(G_{n}\). If \(S = S_{I} \times S_{E}\), where \(S_{I}\) is the
isotropic subgroup and \(S_{E}\) is an entanglement (or symplectic)
subgroup, then using the notations of \cite{lai2013duality}, \(N(S)
= L \times S_{I}\) for some entanglement subgroup \(L\). Let the
sizes of \(S\), \(S_{I}\) and \(S_{E}\) be \(2^{m}\), \(2^{l}\) and
\(2^{2c}\) respectively. Then \(N(S)\) has size \(2^{2n - m}\) and
\(L\) has size \(2^{2n - 2m + 2c}\). According to
\cite{calderbank1998quantum, lai2014dualities}, each \(E \in
\widehat{G}_{n}\) has the form \(E = i^{e} X(a) Z(b)\) with \((a |
b) \in F_{2}^{2n}\), and there is an isometric map \(\tau\) from
\(G_{n}\) to \(F_{2}^{2n}\) such that \(\tau(X(a) Z(b)) = (a | b)\).
Let \(S = \tau(S)\), \(S_{I} = \tau(S_{I})\), \(S_{E} =
\tau(S_{E})\). Then \(S = S_{I} \oplus S_{E}\), \(N(S) = \tau(N(S))
= S^{\perp_{s}}\), and the subspaces \(S\), \(S_{I}\), \(S_{E}\),
\(S^{\perp_{s}}\) are of types \((m, c)_{[s]}\), \((l, 0)_{[s]}\),
\((2c, c)_{[s]}\) and \((2n - m, n + c - m)_{[s]}\), respectively.
The following theorem concerns EA stabilizer codes and the duals of
EAQECCs (its equivalent symplectic formalism can be found in
\cite{guo2013linear}).

\begin{theorem}[\cite{brun2006correcting, lai2014dualities}]
Let \(S = S_{I} \times S_{E}\) and \(N(S)\) be as given above. If
the sizes of \(S\), \(S_{I}\) and \(S_{E}\) are \(2^{m}\), \(2^{l}\)
and \(2^{2c}\) respectively, then
\begin{enumerate}
    \item \(S\) EA-stabilizes an EAQECC \(Q = Q(S) = [[n, k, d_{ea} ; c]]\), where \(k = n + c - m = n - c - l\), \(d_{ea} = \min \{wt(g) \mid g \in N(S) \setminus S_{I}\}\). \(S\) is called the EA-stabilizer of \(Q\) and \(N(S)\) is called the EA-normalizer of \(Q\).
    \item \(N(S)\) EA-stabilizes an EAQECC \(Q^{\perp} = Q(N(S)) = [[n, c, d_{ea}^{\perp} ; k]]\), where \(d_{ea}^{\perp} = \min \{wt(g) \mid g \in S \setminus S_{I}\}\). \(Q^{\perp}\) is called the dual of \(Q\), \(N(S)\) is the EA-stabilizer of \(Q^{\perp}\), and \(S\) is the EA-normalizer of \(Q^{\perp}\).
\end{enumerate}
\end{theorem}

A code \(Q = [[n, k, d_{ea} ; c]]\) is pure if there are no
non-identity elements of \(S_{I}\) with weight \(\leq d_{ea}\) and
impure otherwise \cite{lai2014dualities}. However, constructing
codes using Theorem 1 is challenging. Here we restate Theorem 1
using additive codes. Let \(\chi = \phi \circ \tau\), which is an
isometric map from \(G_{n}\) to \(F_{4}^{n}\). Denote \(S(a) =
\chi(S)\), \(S_{I}(a) = \chi(S_{I})\), and \(S_{E}(a) =
\chi(S_{E})\). Then \(S(a) = S_{I}(a) \oplus S_{E}(a)\),
\(\chi(N(S)) = S(a)^{\perp_{t}}\), and the additive codes \(S(a)\),
\(S_{I}(a)\), \(S_{E}(a)\), \(S(a)^{\perp_{t}}\) are of types \((m,
c)_{[t]}\), \((l, 0)_{[t]}\), \((2c, c)_{[t]}\) and \((2n - m, n + c
- m)_{[t]}\), respectively. We can restate Theorem 1 as

\begin{theorem}
If \(C\) is an \((n, 2^{m})\) additive code of type \((m,
c)_{[t]}\), \(C^{\perp_{t}}\) is of type \((2n - m, n + c -
m)_{[t]}\), and \(R_{t}(C) = C \cap C^{\perp_{t}}\) is an \((n,
2^{l})\) additive code, then
\begin{enumerate}
    \item \(C\) EA-stabilizes an EAQECC \(Q = [[n, k, d_{ea} ; c]]\), where \(k = n + c - m = n - c - l\), \(d_{ea} = \min \{wt(g) \mid g \in C^{\perp_{t}} \setminus R_{t}(C)\}\). \(C\) is called the additive EA-stabilizer of \(Q\) and \(C^{\perp_{t}}\) is called the additive EA-normalizer of \(Q\).
    \item \(C^{\perp_{t}}\) EA-stabilizes an EAQECC \(Q^{\perp} = [[n, c, d_{ea}^{\perp} ; k]]\), where \(d_{ea}^{\perp} = \min \{wt(g) \mid g \in C \setminus R_{t}(C)\}\). \(Q^{\perp}\) is called the dual of \(Q\), \(C^{\perp_{t}}\) is the additive EA-stabilizer of \(Q^{\perp}\), and \(C\) is the additive EA-normalizer of \(Q^{\perp}\).
\end{enumerate}
\end{theorem}

In particular, if \(C\) is an \([n, k]_{4}\) linear code of type
\((2k, c)_{[t]}\), then \(C\) can generate two EAQECCs: \(Q = [[n, n
+ c - 2k, d_{ea} ; c]]\) and \(Q^{\perp} = [[n, c, d_{ea}^{\perp} ;
n + c - 2k]]\).

\section{Bounds of EAQECCs} EA-Singleton bound in
\cite{brun2006correcting} says:

 An \([[n, \kappa, \delta; c]]\) entanglement-assisted quantum
error-correcting code (EAQECC) satisfies
\[
\kappa \leq c + n - 2\delta + 2 \tag{1}
\]
This bound holds for all pure EAQECCs and all EAQECCs with \(\delta
- 1 \leq n/2\) \cite{lai2018linear}, but fails for some impure ones
with \(\delta - 1 \geq n/2\) \cite{grassl2021entanglement}.

The EA-Singleton bounds for an \([[n, \kappa, \delta; c]]\) EAQECC
\(\mathcal{Q}\) presented in \cite[Corollary 9]{grassl2022entropic}
are:
\[
\kappa \leq c + \max\{0, n - 2\delta + 2\} \tag{1'}
\]
\[
\kappa \leq n - \delta + 1 \tag{2}
\]
\[
\kappa \leq \frac{(n - \delta + 1)(c + 2\delta - 2 - n)}{3\delta - 3
- n}, \quad \text{if } \delta - 1 \geq n/2 \tag{3}
\]

To our knowledge, most known families of EAQECCs can achieve bound
(1) when \(\delta - 1 \leq n/2\). Additionally, some EAQECCs with
\(\delta - 1 \geq n/2\) have been constructed from classical MDS
codes (see \cite{grassl2022bounds, luo2022much, huber2020quantum}
and references therein). Focusing on additive EAQECCs constructed
via Theorem 1, whose dimension \(\kappa\) is an integer-bound. Thus
(3) can be rewritten as:
\[
\kappa \leq \left\lfloor \frac{(n - \delta + 1)(c + 2\delta - 2 -
n)}{3\delta - 3 - n} \right\rfloor, \quad \text{if } \delta - 1 \geq
n/2 \tag{3'}
\]

A code with extremal parameters satisfying the EA-Singleton bounds
(1), (2), and (3') is called an EAQMDS code, while a code meeting
bounds (1) and (2) is said to \emph{saturate} the EA-Singleton
bound.

\begin{example}
From the MDS linear codes \([5, 2, 4]_{4}\), \([n, 1, n]_{4}\) for
even \(n\), and the MDS additive codes \((7, 2^{3}, 6)\) and \((8,
2^{5}, 6)\) given in \cite{blokhuis2004small, guo2017construction},
using these codes as additive EA-normalizers, one can obtain the
[[5, 0, 4; 1]], \([[n, 0, n ; n - 2]]\) codes for even \(n\), and
the [[7, 0, 6; 4]] and [[8, 0, 6; 3]] EAQMDS codes from Theorem 2.
\end{example}

\begin{example}
Non-existence of [[5, 2, 4; 3]]. If \(Q = [[5, 2, 4; 3]]\) exists,
its additive EA-normalizer is a \((5, 2^{4})\) additive code of type
\((m, c)_{[t]} = (4, 2)_{[t]}\) by Theorem 2, which must be an MDS
code. According to \cite{ball2025griesmer}, there is only one \((5,
2^{4}, 4)\) MDS code up to equivalence, and this code is of type
\((4, 0)_{[t]}\), which is a contradiction. Thus, the known [[5, 2,
3; 3]] EAQECC in \cite{grassl2022bounds} is optimal.
\end{example}
\section{Constructions of EAQECCs} To construct \([[n, \kappa,
\delta ; c]]\) codes with \(\delta - 1 \geq n/2\) and \(c \geq 1\),
let \(0_{m} = (0, 0, \cdots, 0)\) and \(1_{m} = (1, \cdots, 1)\)
denote the all-zero and all-one vectors of length \(m\),
respectively. For a linear code \(C = [n, k]_{4}\) with generator
matrix \(G = (a_{i,j})\) of size \(k \times n\), its
additive generator matrix is \(G_{a} = \begin{pmatrix} \omega G \\
\varpi G \end{pmatrix}\) of size \(2k \times n\). We use \(G =
(a_{i,j})_{[L]}\) and \(B = (b_{i,j})_{[A]}\) to denote that \(G\)
is the generator matrix of a quaternary linear code and \(B\) is an
additive generator matrix of an additive code, respectively.
According to Theorem 2,  we can obtain the following result based on
classical additive codes.

\begin{theorem}
For \(m \geq 0\), the following EAQECCs saturate the EA-Singleton
bound:
\begin{enumerate}
    \item If \(n \geq 4\), there exists an \([[n, 1, n - 1 ; n - 3]]\) code.
    \item If \(n \geq 5\) is odd, there exists an \([[n, 1, n - 2 ; n - 5]]\) code.
    \item If \(s \geq 1\) and \(n = 8s + 1 + 2m \geq 9\), there exists an \([[n, 1, n - 2s ; n - 4s - 1]]\) code.
    \item If \(s \geq 1\) and \(n = 8s + 4 + 2m \geq 12\), there exists an \([[n, 1, n - 2s - 1 ; n - 4s - 3]]\) code.
\end{enumerate}
\end{theorem}

\begin{proof}
\begin{enumerate}
    \item If \(n = 4 + 2m \geq 4\) is even, let
    \[G_{2,4} = \begin{pmatrix} 11110 \\ 01\omega\varpi \end{pmatrix}_{[L]}, \quad G_{2,n} = \begin{pmatrix} 1111 | 0_{2m} \\ 01\omega\varpi | 1_{2m} \end{pmatrix}_{[L]}\]

    \(G_{2,4}\) generates a \([4, 2, 3]\) linear code of type \((4, 1)_{[t]}\), and \(G_{2,n}\) generates a \(C_{n} = [n, 2, 4]\) linear code of type \((4, 1)_{[t]}\). The weight enumerator of \(C_{n}\) is \(W(t) = 1 + 3z^{4} + 12z^{n-1}\) and that of \(C_{n} \cap C_{n}^{\perp_{t}}\) is \(W_{R}(t) = 1 + 3z^{4}\). Thus, \(C_{n}\) normalizes an \([[n, 1, n - 1 ; n - 3]]\) EAQECC.

    If \(n = 5 + 2m > 5\) is odd, let
    \[G_{4,5} = \begin{pmatrix} 11000 \\ 00110 \\ \omega\varpi011 \\ 01\varpi\omega \end{pmatrix}_{[A]}, \quad G_{4,n} = \begin{pmatrix} 11000 & 0_{2m} \\ 00110 & 0_{2m} \\ \omega\varpi011 & 1_{2m} \\ 01\varpi\omega & \omega1_{2m} \end{pmatrix}_{[A]}\]

    \(G_{4,5}\) generates a \((5, 2^{4}, 3)\) additive code of type \((4, 1)_{[t]}\), and \(G_{4,n}\) generates a \(C_{n} = (n, 2^{4})\) additive code of type \((4, 1)_{[t]}\). The weight enumerator of \(C_{n}\) is \(W(t) = 1 + 2z^{2} + z^{4} + 12z^{n-1}\) and that of \(C_{n} \cap C_{n}^{\perp_{t}}\) is \(W_{R}(t) = 1 + 2z^{2} + z^{4}\). Thus, \(C_{n}\) normalizes an \([[n, 1, n - 1 ; n - 3]]\) EAQECC.

    \item If \(n = 5 + 2m \geq 5\) is odd, let \(G_{3,5}\) generate a \([5, 3, 3]_{4}\) linear code of type \((6, 1)_{[t]}\), where
    \[G_{3,5} = \begin{pmatrix} 10111 \\ 011\omega\varpi \\ 00\varpi\omega1 \end{pmatrix}_{[L]}, \quad G_{3,n} = \begin{pmatrix} 10111 | 0_{2m} \\ 011\omega\varpi | 0_{2m} \\ 00\varpi\omega1 | 1_{2m} \end{pmatrix}_{[L]}\]

    \(G_{3,n}\) generates a \(C_{n} = [n, 3, 4]_{4}\) linear code of type \((6, 1)_{[t]}\). The weight enumerator of \(C_{n}\) is \(W(t) = 1 + 15z^{4} + 48z^{n-2}\) and that of \(C_{n} \cap C_{n}^{\perp_{t}}\) is \(W_{R}(t) = 1 + 15z^{4}\). Thus, \(C_{n}\) normalizes an \([[n, 1, n - 2 ; n - 5]]\) EAQECC.

    \item If \(n = 8s + 1 + 2m \geq 9\) is odd, let \(A = 1_{4}\), \(B = 0_{4}\), \(D = (0, 1, \omega, \varpi)\), and construct the \((2s + 1) \times n\) matrix
    \[G_{2s+1,n} = \begin{pmatrix} AB\cdots B & 0_{2m+1} \\ BA\cdots B & 0_{2m+1} \\ \vdots & \vdots \\ BB\cdots A & 0_{2m+1} \\ DD\cdots D & 1_{2m+1} \end{pmatrix}_{[L]}\]

    \(G_{2s+1,n}\) generates a \(C_{n} = [n, 2s + 1, 4]\) linear code of type \((2(2s + 1), 1)_{[t]}\), and its first \(2s\) rows generate \(R(C_{n}) = C_{n} \cap C_{n}^{\perp_{t}}\). The weight enumerator of \(R(C_{n})\) is \(W_{R}(t) = 1 + a_{4}z^{4} + a_{8}z^{8} + \cdots + a_{8s}z^{8s}\), and that of \(C_{n}\) is \(W(t) = W_{R}(t) + 3 \times 4^{2s}z^{n-2s}\). Thus, \(C_{n}\) normalizes an \([[n, 1, n - 2s ; n - 4s - 1]]\) EAQECC.

    \item If \(n = 8s + 4 + 2m \geq 12\) is even, let \(A = 1_{4}\), \(B = 0_{4}\), \(D = (0, 1, \omega, \varpi)\), and construct the \((2s + 2) \times n\) matrix
    \[G_{2s+2,n} = \begin{pmatrix} AB\cdots B & 0_{2m} \\ BA\cdots B & 0_{2m} \\ \vdots & \vdots \\ BB\cdots A & 0_{2m} \\ DD\cdots D & 1_{2m} \end{pmatrix}_{[L]}\]

    \(G_{2s+2,n}\) generates a \(C_{n} = [n, 2s + 2, 4]\) linear code of type \((2(2s + 2), 1)_{[t]}\), and its first \(2s + 1\) rows generate \(R(C_{n}) = C_{n} \cap C_{n}^{\perp_{t}}\). The weight enumerator of \(R(C_{n})\) is \(W_{R}(t) = 1 + a_{4}z^{4} + a_{8}z^{8} + \cdots + a_{8s+4}z^{8s+4}\), and that of \(C_{n}\) is \(W(t) = W_{R}(t) + 3 \times 4^{2s+1}z^{n-2s-1}\). Thus, \(C_{n}\) normalizes an \([[n, 1, n - 2s - 1 ; n - 4s - 3]]\) EAQECC.
\end{enumerate}
\end{proof}

It is easy to verify that all these EAQECCs saturate the
EA-Singleton bound, and the codes in class (1) are EAQMDS codes.
Except for the [[4, 1, 3; 1]] and [[5, 1, 3; 0]] codes, the others
are new and impure.
\section{Conclusion}
We have established an additive EA stabilizer formalism for EAQECCs
and their duals through the induced link between the geometry of the
symplectic group and EAQECCs, which is equivalent to the formalisms
in \cite{brun2006correcting, lai2014dualities}. This formalism
enables researchers to easily construct EAQECCs from any classical
additive code and may be used to derive sharp bounds for additive
EAQECCs and analyze their optimality. Moreover, this formalism can
be generalized to non-binary EAQECCs using known formalisms in
\cite{grassl2022bounds, luo2022much, galindo2019entanglement} and
group theory in \cite{wan1993geometry}.

We have proposed constructions of many good EAQECCs, disproven a
conjecture about EAQECCs, and illustrated the process of analyzing
the optimality of EAQECCs with an example. Based on
\cite{blokhuis2004small, guo2017construction, ball2025griesmer,
guan2023some}, we have constructed over 60 optimal EAQECCs and some
EAQECCs with better parameters than the best-known ones in
\cite{grassl2022entropic}, which will be presented in
\cite{liu2023good}. The additive EA stabilizer formalism is also
useful for designing encoders and decoders, as demonstrated in
\cite{lai2014dualities}, and may be applied to study physically
realizable high-performance EAQECCs. These topics will be
interesting directions for future research in quantum computation
and quantum information.
\section*{Acknowledgment}
This work is supported by  National Natural Science Foundation of
China under Grant No.U21A20428 and Natural Science Foundation of
Shaanxi Province under Grant Nos.2024JC-YBMS-055 and
2025JC-YBQN-070.


\begin{thebibliography}{1}
 \bibliographystyle{IEEEtran}
    \bibitem{shor1995scheme} P. W. Shor, ``Scheme for reducing decoherence in quantum computer memory,'' \emph{Physical Review A}, vol. 52, no. 4, p. R2493, 1995.
    \bibitem{steane1996error} A. M. Steane, ``Error correcting codes in quantum theory,'' \emph{Physical Review Letters}, vol. 77, no. 5, pp. 793--797, 1996.
    \bibitem{gottesman1997stabilizer} D. Gottesman, ``Stabilizer codes and quantum error correction,'' Ph.D. dissertation, California Inst. Technol., Pasadena, CA, USA, 1997.
    \bibitem{calderbank1998quantum} A. R. Calderbank, E. M. Rains, P. M. Shor, and N. J. A. Sloane, ``Quantum error correction via codes over GF(4),'' \emph{IEEE Transactions on Information Theory}, vol. 44, no. 4, pp. 1369--1378, 1998.
    \bibitem{ketkar2006nonbinary} A. Ketkar, A. Klappenecker, S. Kumar, and P. K. Sarvepalli, ``Nonbinary stabilizer codes over finite fields,'' \emph{IEEE Transactions on Information Theory}, vol. 52, no. 11, pp. 4892--4914, 2006.
    \bibitem{grassl2006bounds} M. Grassl, ``Bounds on the minimum distance of linear codes and quantum codes,'' [Online]. Available: http://codetables.de/.
    \bibitem{mackay2004sparse} D. J. C. MacKay, G. Mitchison, and P. L. McFadden, ``Sparse-graph codes for quantum error-correction,'' \emph{IEEE Transactions on Information Theory}, vol. 50, no. 10, pp. 2315--2330, 2004.
    \bibitem{bowen2002entanglement} G. Bowen, ``Entanglement required in achieving entanglement-assisted channel capacities,'' \emph{Physical Review A}, vol. 66, no. 5, p. 052313, 2002.
    \bibitem{brun2006correcting} T. Brun, I. Devetak, and M.-H. Hsieh, ``Correcting quantum errors with entanglement,'' \emph{Science}, vol. 314, no. 5798, pp. 436--439, 2006.
    \bibitem{wilde2008optimal} M. M. Wilde and T. A. Brun, ``Optimal entanglement formulas for entanglement-assisted quantum coding,'' \emph{Physical Review A}, vol. 77, no. 6, p. 064302, 2008.
    \bibitem{lai2013entanglement} C.-Y. Lai and T. A. Brun, ``Entanglement increases the error-correcting ability of quantum error-correcting codes,'' \emph{Physical Review A}, vol. 88, no. 1, p. 012320, 2013.
    \bibitem{hsieh2011high} M. H. Hsieh, W. T. Yen, and L. Y. Hsu, ``High performance entanglement-assisted quantum LDPC codes need little entanglement,'' \emph{IEEE Transactions on Information Theory}, vol. 57, no. 3, pp. 1761--1769, 2011.
    \bibitem{li2014entanglement} R. Li, L. Guo, and Z. Xu, ``Entanglement-assisted quantum codes achieving the quantum Singleton bound but violating the quantum Hamming bound,'' \emph{Quantum Information and Computation}, vol. 14, nos. 13--14, pp. 1107--1116, 2014.
    \bibitem{lai2013duality} C.-Y. Lai, T. A. Brun, and M. M. Wilde, ``Duality in entanglement-assisted quantum error correction,'' \emph{IEEE Transactions on Information Theory}, vol. 59, no. 6, pp. 4020--4024, 2013.
    \bibitem{lai2014dualities} C.-Y. Lai, T. A. Brun, and M. M. Wilde, ``Dualities and identities for entanglement-assisted quantum codes,'' \emph{Quantum Information Processing}, vol. 13, pp. 957--990, 2014.
    \bibitem{lai2018linear} C.-Y. Lai and A. Ashikhmin, ``Linear programming bounds for entanglement-assisted quantum error-correcting codes by split weight enumerators,'' \emph{IEEE Transactions on Information Theory}, vol. 64, no. 1, pp. 622--639, 2018.
    \bibitem{guo2013linear} L. Guo and R. Li, ``Linear Plotkin bound for entanglement-assisted quantum codes,'' \emph{Physical Review A}, vol. 87, no. 3, p. 032309, 2013.
    \bibitem{grassl2021entanglement} M. Grassl, ``Entanglement-assisted quantum communication beating the quantum Singleton bound,'' \emph{Physical Review A}, vol. 103, no. 2, p. L020601, 2021.
    \bibitem{grassl2022entropic} M. Grassl, F. Huber, and A. Winter, ``Entropic proofs of Singleton bounds for quantum error-correcting codes,'' \emph{IEEE Transactions on Information Theory}, vol. 68, no. 6, pp. 3942--3950, 2022.
    \bibitem{grassl2022bounds} M. Grassl, ``Bounds on the minimum distance of entanglement-assisted quantum codes,'' [Online]. Available: http://codetables.de/EAQECC/.
    \bibitem{luo2022much} G. Luo, M. F. Ezerman, M. Grassl, and S. Ling, ``How much entanglement does a quantum code need?'' arXiv:2207.05647, 2022.
    \bibitem{huber2020quantum} F. Huber and M. Grassl, ``Quantum codes of maximal distance and highly entangled subspaces,'' \emph{Quantum}, vol. 4, p. 284, 2020.
    \bibitem{galindo2019entanglement} C. Galindo, F. Hernando, R. Matsumoto, and D. Ruano, ``Entanglement-assisted quantum error-correcting codes over arbitrary finite fields,'' \emph{Quantum Information Processing}, vol. 18, no. 4, p. 116, 2019.
    \bibitem{wan1993geometry} Z. Wan, \emph{Geometry of classical groups over finite fields and its applications}. Lund, Sweden: Chart Well Bratt, 1993.
    \bibitem{blokhuis2004small} A. Blokhuis and A. E. Brouwer, ``Small additive quaternary codes,'' \emph{European Journal of Combinatorics}, vol. 25, no. 2, pp. 161--167, 2004.
    \bibitem{guo2017construction} L. B. Guo, Y. Liu, L. D. Lu, and R. H. Li, ``On construction of good quaternary additive codes,'' in \emph{ITM Web of Conferences}, vol. 12, 2017, p. 03013.
    \bibitem{ball2025griesmer} S. Ball, M. Lavrauw, and T. Popatia, ``Griesmer type bounds for additive codes over finite fields, integral and fractional MDS codes,'' \emph{Designs, Codes and Cryptography}, vol. 93, pp. 175--196, 2025.
    \bibitem{guan2023some} C. Guan, R. Li, Y. Liu, and Z. Ma, ``Some quaternary additive codes outperform linear counterparts,'' \emph{IEEE Transactions on Information Theory}, vol. 69, no. 11, pp. 7122--7131, 2023.
    \bibitem{liu2023good} Y. Liu and C. Guan, ``Good additive EAQECCs codes from short additive quaternary codes,'' in preparation.
\end{thebibliography}
\end{document}